\documentclass[12pt]{article}
\usepackage{amsmath}
\usepackage{amssymb}
\usepackage{amsthm}
\usepackage{fullpage}
\usepackage[utf8]{inputenc}
\usepackage{mathtools}
\usepackage{url}
\usepackage[usenames,svgnames]{xcolor}
\usepackage{cite}

\usepackage{fullpage}
\usepackage[utf8]{inputenc}
\usepackage{mathtools}
\usepackage{url}
\usepackage[usenames,svgnames]{xcolor}
\usepackage{cite}
\usepackage{enumi tem}
\usepackage{datetime}
\usepackage[titletoc,title]{appendix}
\usepackage{comment}
\usepackage{mathdots}
\usepackage{graphicx}

\usepackage[pagebackref=false]{hyperref} 
\usepackage[all]{hypcap} 
\RequirePackage{color}

\usepackage{amsmath,amssymb,amsthm,tikz}
\usetikzlibrary{decorations.pathreplacing}
\usepackage{bm}
\usepackage{caption}
\usepackage[final,color]{showkeys}
\usepackage{xr}

\usepackage{mdwlist}	

\usepackage{placeins}

\usepackage{nicefrac}	

\theoremstyle{plain}
\newtheorem{theorem}{\sc Theorem}[section]

\newtheorem{corollary}[theorem]{\sc Corollary}
\newtheorem{example}{Example}[section]
\newtheorem{examples}{Example}[subsection]

\newtheorem{remark}{Remark}[section]

\theoremstyle{definition}
\newtheorem{definition}{Definition}[section]

\numberwithin{equation}{section} 

\numberwithin{equation}{section}
\theoremstyle{remark}

\newcommand{\pmeas}[2]{\xi_{#2}^{(#1)}}

\newcommand{\pmp}{\pmeas{d}{E'(q)}}

\newcommand{\Pmeas}[3]{\theta_{#3}^{(#2,#1)}}
\newcommand{\Pmp}{\Pmeas{d}{n}{E'(q)}}

\newcommand{\pf}[2]{Z_{#2}^{(#1)}}
\newcommand{\pfp}{\pf{d}{E'(q)}}

\newcommand{\ip}[1]{\left\langle #1 \right\rangle}

\DeclareMathOperator{\cyc}{cyc}

\DeclareMathOperator{\aut}{aut}

\def\ra{{\rightarrow}}

\definecolor{my-blue}{rgb}{0.0,0.0,0.6}
\definecolor{my-red}{rgb}{0.5,0.0,0.0}
\definecolor{my-green}{rgb}{0.0,0.5,0.0}
\definecolor{nicos-red}{rgb}{0.75,0.0,0.0}
\definecolor{light-gray}{gray}{0.6}
\definecolor{really-light-gray}{gray}{0.8}

\def\be{\begin{equation}}
\def\ee{\end{equation}}

\def\bea{\begin{eqnarray}}
\def\eea{\end{eqnarray}}

\def\bt{\begin{theorem}}
\def\et{\end{theorem}}

\def\bex{\begin{example}\small \rm}
\def\eex{\end{example}}

\def\bexs{\begin{examples}\small \rm}
\def\eexs{\end{examples}}

\def\ra{\rightarrow}

\def\deq{\coloneqq}

\def\br{\begin{remark}\small \rm}
\def\er{\end{remark}}

\def\&{&{\hskip -20pt}}

\def\Ib{\mathbf{I}}

\def\Pb{\mathbf{P}}

\def\Zbb{\mathbb{Z}}

\def\Zbb{\mathbb{Z}}

\let\Oldsection\section
\renewcommand{\section}{\FloatBarrier\Oldsection}

\let\Oldsubsection\subsection
\renewcommand{\subsection}{\FloatBarrier\Oldsubsection}

\let\Oldsubsubsection\subsubsection
\renewcommand{\subsubsection}{\FloatBarrier\Oldsubsubsection}

\def\mP{\mathcal{P}}

\newcommand{\rb}[1]{\left(#1\right)}

\newcommand{\abs}[1]{\left|#1\right|}

\newcommand{\set}[1]{\left\{#1\right\}}

\newcommand{\lambdamap}{\Lambda^{(n)}_d}

\newcommand{\rr}{\mathbb{R}}

\definecolor{darkgreen}{rgb}{0.0,0.5,0.0}
\definecolor{darkblue}{rgb}{0.0,0.0,0.3}
\definecolor{nicosred}{rgb}{0.65,0.1,0.1}
\definecolor{light-gray}{gray}{0.7}
\allowdisplaybreaks[1]

\newcommand{\fM}{\mathfrak{M}}

\begin{document}
\baselineskip 16pt
\begin{flushright}
CRM 3357 (2016)
\end{flushright}
\medskip
\begin{center}
\begin{Large}\fontfamily{cmss}
\fontsize{17pt}{27pt}
\selectfont
\textbf{Zero-temperature limit of \\ quantum weighted Hurwitz numbers}
\end{Large}\\
\bigskip
\begin{large}  {J. Harnad}$^{1,2}$ and {Janosch Ortmann}$^{1,2}$ 
 \end{large}
\\
\bigskip
\begin{small}
$^{1}${\em Centre de recherches math\'ematiques,
Universit\'e de Montr\'eal\\ C.~P.~6128, succ. centre ville, Montr\'eal,
QC, Canada H3C 3J7 } \\
\smallskip
$^{2}${\em Department of Mathematics and
Statistics, Concordia University\\ 1455 de Maisonneuve Blvd.~W.  
Montr\'eal, QC,  Canada H3G 1M8 } 
\end{small}
\end{center}
\bigskip

\begin{abstract}
   The partition function  for quantum weighted double Hurwitz numbers can be interpreted in 
   terms of the energy distribution of a quantum Bose gas with vanishing fugacity.
We compute the leading term of the partition function  and the quantum weighted 
   Hurwitz numbers  in the zero temperature limit $T \ra 0$,  as well as  the next order corrections. 
   The leading term  is shown to reproduce the case of uniformly weighted Hurwitz numbers of Belyi curves. 
   In particular, the KP or Toda $\tau$-function serving  as generating function for the quantum Hurwitz
numbers is shown in the limit to give the one for Belyi curves  and, with
suitable scaling,  the same holds true for the partition function, the weights and the expectations of Hurwitz numbers.

     \end{abstract}



\section{Introduction: quantum weighted Hurwitz numbers and their generating functions}
\label{sec:intro}

\subsection{Hurwitz numbers}
\label{subsev:hurwitz}

Multiparametric weighted  Hurwitz numbers  were  introduced in \cite{GH1, GH2, H1, HO} as a generalization
of the notion of simple Hurwitz numbers \cite{Hu1, Hu2, Pa, Ok} and  other previously studied special cases 
 \cite{GGN, AC1, AC2, AMMN1, AMMN2, KZ, Z}. It was shown  that parametric families of
 KP or $2D$ Toda $\tau$-functions of {\it hypergeometric type} \cite{KMMM, OrSc} serve as generating 
 functions for all such weighted Hurwitz numbers, with the latter appearing as coefficients  in a single or double expansion
 over the basis of power sum symmetric functions in an auxiliary set of variables.
  The weights are determined by a weight  generating functions $G(z)$, which
  depends  on a possibly infinite set of parameters ${\bf c}=(c_1, c_2, \dots)$.
This can either be expressed as a formal  sum
 \be
 G(z) = 1 + \sum_{i=1}^\infty g_i z^i
  \label{G_weight_gen_sum}
 \ee
 or an infinite product
 \be
 G(z) = \prod_{i=1}^\infty (1 +z c_i),
 \label{G_weight_gen_prod}
 \ee
 or some limit thereof. Comparing the two, $G(z)$ can be viewed as the generating
 function for elementary symmetric functions in the variables ${\bf c} =(c_1, c_2, \dots)$.
 \be
 g_i = e_i({\bf c}).
 \ee
 
For a set of $k$ partitions $(\mu^{(1)}, \cdots , \mu^{(k)})$ of $n$,  the geometrical definition of  the pure Hurwitz number  $H(\mu^{(1)}, \cdots , \mu^{(k)})$  is:
 \begin{definition}
 $H(\mu^{(1)}, \cdots , \mu^{(k)})$
 is the number of distinct $n$-sheeted branched coverings $\Gamma \ra \Pb^1$ of the Riemann sphere
 having $k$ branch points $(p_1, \dots , p_k)$ with ramification profiles $\{\mu^{(i)}\}_{i=1, \dots , k}$,
 divided by the order $\aut(\Gamma)$ of the automorphism group of $\Gamma$.
 \end{definition}
An  equivalent combinatorial definition  \cite{Hu1, Hu2, Frob1, Frob2, LZ} is:
 \begin{definition}  $H(\mu^{(1)}, \cdots , \mu^{(k)})$ is the number of distinct factorization
 of the identity element $\Ib \in S_n$ of the symmetric group as an ordered product 
 \be
 \Ib = h_1, \dots h_k, \quad h_i \in S_n, \quad i=1, \dots , k
\ee
where $h_i$ belongs to the conjugacy class $\cyc(\\mu^{(i)}$ with cycle lengths equal to the parts 
of $\mu^{(i)}$, divided by $n!$.
\end{definition}

We denote the   weight of a partition $| \mu|$ , its length $\ell(\mu)$ and define its {\it colength} as
\be
\ell^*(\mu):= |\mu| - \ell(\mu).
\ee

\subsection{Weighted Hurwitz numbers}
\label{subsec:weighted}

As in \cite{GH1, GH2, H1, HO}  for each positive integer $d$, 
and every pair of ramification profiles $(\mu, \nu)$ (i.e. partitions of $n$),
we define the weighted double Hurwitz number $H^d_G(\mu, \nu)$ as
   \be
H^d_G(\mu, \nu) := \sum_{k=0}^\infty \sideset{}{'}\sum_{\substack{\mu^{(1)}, \dots \mu^{(k)} \\ \sum_{i=1}^k \ell^*(\mu^{(i)})= d}}
m_\lambda ({\bf c})H(\mu^{(1)}, \dots, \mu^{(k)}, \mu, \nu) ,
\label{Hd_G}
\ee
where 
\be
m_\lambda ({\bf c}) :=
\frac{1}{\abs{\aut(\lambda)}} \sum_{\sigma\in S_k} \sum_{1 \le i_1 < \cdots < i_k}
 c_{i_\sigma(1)}^{\lambda_1} \cdots c_{i_\sigma(k)}^{\lambda_k},
 \label{monomial_sf}
\ee
is the monomial sum symmetric  function \cite{Mac} corresponding to a partition $\lambda$ of weight 
\be
|\lambda|= d = \sum_{i=1}^k \ell^*(\mu^{(i)})
\label{d_def}
\ee
whose parts $\{\lambda_i\}$ are the colengths $\{\ell^*(\mu^{(i)})\}$ in weakly descending order,
\be
|\aut (\lambda)| := \prod_{i=1}^{\ell(\lambda)} m_i(\lambda)!,
\ee
where $m_i(\lambda)$ is the number of parts of $\lambda$ equal to $i$
and $\sum'$ denotes the sum over all $k$-tuples of partitions $(\mu^{(1)}, \dots, \mu^{(k)})$ 
 satisfying condition (\ref{d_def}) other than the cycle type of the identity element.

The Euler characteristic of the covering surface is given by the Riemann-Hurwitz formula, 
\be
\chi = 2-2g = 2n - d.
\ee

The particular case where all the $\mu^{(i)}$'s represent simple branching was
 studied from this viewpoint  in \cite{Pa, Ok}. It corresponds to the exponential weight generating function
\be
G(z) = e^z.
\ee
The weight is uniform in this case  on all $k$-tuples $(\mu^{(1)}, \cdots, \mu^{(k)})$ of  partitions corresponding 
to simple branching
\be
\mu^{(i)} = (2, (1)^{n-2})
\ee
and vanishes on all others.  This is what is referred to as ``simple''  (double) Hurwitz numbers.

\subsection{The $\tau$-function as generating function}
\label{subsec:taufn}

Choosing a small parameter $\beta$, the following double Schur function  expansion defines 
the $N=0$ value of a $2D$-Toda $\tau $ function of hypergeometric type (at the lattice point $N=0$).
\be
   \tau^{(G, \beta)} ({\bf t}, {\bf s})  = \sum_{\lambda} \ r^{(G, \beta)}_\lambda s_\lambda ({\bf t}) s_\lambda ({\bf s}),
   \label{tau_G}
  \ee
  where the coefficients $r^{(G, \beta)}_\lambda$  are determined by the weight generating
  function $G$ by the following {\it content product} formula
  \be
r_\lambda^{(G, \beta)} :=   \prod_{(i,j)\in \lambda} G(\beta( j-i)),
\label{r_G_lambda}
\ee

By changing the expansion basis from Schur functions  \cite{Mac} to the power sum symmetric 
functions $p_\mu({\bf t}) p_\nu({\bf s})$ it  was proved in \cite{GH1, GH2, H1, HO}, 
that $ \tau^{(G, \beta)} ({\bf t}, {\bf s})) $  serves as a generating function for the weighted double Hurwitz numbers $H^d_G(\mu, \nu)$.
\begin{theorem}
\label{tau_H_G_generating function}
The 2D Toda $\tau$-function $\tau^{(G, \beta)}({\bf t}, {\bf s})$
can be expressed as
\be
\tau^{(G, \beta)}({\bf t}, {\bf s}) =\sum_{d=0}^\infty \beta^d \sum_{\substack{\mu, \nu \\ |\mu|=|\nu|}} H^d_G(\mu, \nu) p_\mu({\bf t}) p_\nu({\bf s}).
\label{tau_H_G}
\ee
\end{theorem}

\subsection{Quantum Hurwitz numbers}
\label{subsec:quantumhurwitz}

By  {\it simple quantum Hurwitz numbers} \cite{GH2, H2}  we mean the special case  of weighted Hurwitz numbers  obtained by choosing
 the parameters $c_i$ as
\be
c_i = q^i, \quad i=1, 2 \dots 
\label{c_iq_i_prime}
\ee
where $q$ is a real parameter between $0$ and $1$.

The corresponding weight generating function is 
\be
G(z) = E'(q,z)  \deq \prod_{i=1}^\infty (1+ q^i z) = (-zq; q)_\infty
\label{E_prime_qz__def}
\ee
where
\be
(z;q)_k := \prod_{j=0}^{k-1}((1 - z q^j), \quad (z; q)_\infty := \prod_{j=0}^{\infty}(1 - z q^j)
\ee
is the quantum Pochhammer symbol. 

Making the substitutions (\ref{c_iq_i_prime}), the weights entering in (\ref{Hd_G}) become
\bea
\label{eq:WePrime}
W_{E'(q)} (\mu^{(1)}, \dots, \mu^{(k)}) &\& := m_\lambda (q, q^2, \dots )\cr
&\& =  {1\over  |\aut(\lambda)|} \sum_{\sigma\in S_k} \frac{1}{
(q^{-\ell^*(\mu^{(\sigma(1))})} -1) \cdots (q^{-\ell^*(\mu^{(\sigma(1))})} \cdots q^{-\ell^*(\mu^{(\sigma(k))})}-1)}, \cr
&\&
\label{W_Eprime_q}
\eea
The (unnormalized) weighted Hurwitz numbers  therefore become
  \be
H^d_{E'(q)}(\mu, \nu) := \sum_{k=0}^\infty \sideset{}{'}\sum_{\substack{\mu^{(1)}, \dots \mu^{(k)} \\ \sum_{i=1}^k \ell^*(\mu^{(i)})= d}}
W_{E'(q)} (\mu^{(1)}, \dots, \mu^{(k)}) H(\mu^{(1)}, \dots, \mu^{(k)}, \mu, \nu). 
\label{Hd_E_prime_q}
\ee

\br
The parameter $q$ may be interpreted as $q = e^{-\epsilon}$ for a small parameter 
\be
\epsilon = { E_0 \over k_B T}, \quad T= \text{temperature}, \quad k_B =\text{Boltzman constant},
\ee
where 
\be
E_0 = \hbar \omega_0
\ee
is  interpreted as the ground state energy, and the higher levels are integer multiples
proportional to the colength of the partition representing the ramification type of a branch point
\be
E(\mu) = \ell^*(\mu) E_0.
\ee
Then (\ref{W_Eprime_q}) may be interpreted in terms of the energy distribution of a quantum bose gas with vanishing fugacity
\be
n_E = {1 \over e^{ E\over k_B T} -1},
\ee
assuming that the energy for $k$ branch points is the sum of that for each
\be
E(\mu^{(1)}, \dots, \mu^{(k)}) = \sum_{i=1}^k \ell^*(\mu^{(i)})\hbar \omega_0.
\ee
\er

Choosing $G= E'(q)$, in eqs.~(\ref{tau_G}). (\ref{tau_H_G}), we obtain
\bea
   \tau^{(G, \beta)} ({\bf t}, {\bf s})  &\&= \sum_{\lambda} \ r^{(E'(q), \beta)}_\lambda(z)s_\lambda ({\bf t}) s_\lambda ({\bf s}) \cr
   &\& =\sum_{d=0}^\infty \beta^d \sum_{\substack{\mu, \nu \\ |\mu|=|\nu|}} H^d_{E'(q)}(\mu, \nu) p_\mu({\bf t}) p_\nu({\bf s}).
   \label{tau_Eprimeq}
\eea
as the generating function for  simple quantum Hurwitz numbers.


\section{Quantum Hurwitz numbers and probability measure on $\mP_d$}

The quantum weighted Hurwitz numbers can be interpreted in terms of probability measures on the set of integer partitions. We summarize the pertinent facts here and refer to Section 2 of \cite{HOrt} for further details.

For $k\in\set{1, \dots , d}$ consider the sets
\begin{align}
		\label{eq:defM}
		\fM_{d,k}^{(n)} & = \set{\rb{\mu^{(1)}, \ldots, \mu^{(k)}} \in \rb{\mP_n}^k \colon \ell^*\rb{\mu^{(j)}}\ne 0 \ \forall j,\,  \sum_{j=1}^k \ell^\ast \rb{\mu^{(j)}} =d } \quad\text{and}\quad
		\fM_d^{(n)}&=\coprod_{k=1}^{d} \fM^{(n)}_{d,k}.
\end{align}
Define a measure $\Pmp$ on $\fM^{(n)}_{d}$ by
\begin{align}
	\label{eq:defTheta}
	\Pmp \rb{\mu^{(1)}, \ldots, \mu^{(k)}}  & = \frac1{\pfp} W_{E'(q)}\rb{ \mu^{(1)},\ldots,\mu^{(k)} },
	\intertext{where the \emph{partition function} $\pfp$ is defined so that $\Pmp$ is a probability measure; that is,}
	\label{eq:defZt}
	\pfp &= \sum_{k=1}^d \sum_{\fM^{(n)}_{d,k}} W_{E'(q)}\rb{\mu^{(1)},\ldots,\mu^{(k)} }.
	\intertext{The fact that $\pfp$ does not depend on $n$, provided $n\geq 2d$, is a consequence of the results in Section 2 of \cite{HOrt}. We then have the expectation value}
	\label{eq:wHexp}
	\ip{H\rb{\cdot,\ldots,\cdot,\mu,\nu}}_{\Pmp} & = \frac1{\pfp} H_{E'(q)}^d(\mu,\nu),
\end{align}
where $\ip{\cdot}_{\Pmp}$ denotes integration with respect to the measure $\Pmp$.

\begin{definition}
	For $n,d\in\Zbb_{>0}$ define the function $\lambdamap\colon \fM_d^{(n)}\longrightarrow \mP_d$ as follows:
\begin{align}
	\lambdamap & \colon \rb{\mu^{(1)},\ldots,\mu^{(k)}}\longmapsto \lambda
	\intertext{where $\lambda$ is the unique partition of $d$ such that}
	\set{\lambda_1,\ldots,\lambda_k} & = \set{ \ell^\ast\rb{\mu^{1} },\ldots, \ell^\ast\rb{\mu^{(k)} }  }.
\end{align}
\end{definition}

Assume from now on that $n\geq 2d$. By Lemma 2.1 of \cite{HOrt} the push-forward $\pmp$ on $\mP_d$ of $\Pmp$ under $\lambdamap$ does not depend on $n$ and is given by
\begin{align}
	\pmp(\lambda) &= \frac1{\pfp}\, p\rb{\lambda} w_{E'(q)}(\lambda)  \quad\quad\forall\, \lambda\in\mP_d,
	\intertext{where for any $\lambda\in\mP_d$,}
	w_{E'(q)}(\lambda) & = \frac1{\abs{\aut(\lambda)}}\, \sum_{\sigma\in S_{\ell(\lambda)}} \prod_{j=1}^{\ell(\lambda)} \rb{q^{\sum_{i=1}^j \lambda_{\sigma(i)} }}^{-1}
\end{align}
and further
\begin{align}
	p(\lambda) & := \prod_{j=1}^{\ell(\lambda)} p(\lambda_j).
\end{align}


\section{Zero-temperature limit and asymptotic expansion}

\label{sec:scl}

Recalling that the parameter $q$ is interpreted as 
\be
	q  = e^{-{ E_0 \over k_B T}}
\ee
for some ground state energy $E_0=\hbar\omega_0$, the zero temperature limit $T\longrightarrow 0^+$ corresponds to $q\longrightarrow 0^+$. In this section we state our asymptotic results in this limit. All proofs are given in the following section. 
We further assume throughout that $d\geq 2$ and $n\geq 2d$.


\subsection{Zero-temperature limit: leading term}
 \begin{definition}
	The \emph{Dirac measure} $\delta_x$ at $x\in S$ on a measurable space $(S,\Sigma)$ is defined by
	\begin{align}
		\label{eq:defDirac}
		\delta_x(A) & = \begin{cases}
			1\quad&\text{if } x\in A\\
			0& \text{otherwise}
		\end{cases}
	\end{align}
	for all $A\in\Sigma$.
 \end{definition}
 We begin by stating the leading order zero temperature limit.

\begin{theorem}
	\label{thm:Downstairs}
	Let $d\in\Zbb_{>0}$. As $q\longrightarrow 0^+$, the  $1$-parameter family of measures $\rb{\pmp}$ on $\mP_d$ converges weakly to the Dirac measure $\delta_{(d)}$ at $(d)\in\mP_d$.
\end{theorem}

\ 

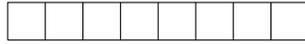
\begin{figure}[ht]

	\begin{center}
	\begin{tikzpicture}[>=latex,scale=0.5]

		\draw (0,0)--(1,0)--(2,0)--(3,0)--(4,0) -- (5,0) -- (6,0) --(7,0) --(8,0) -- (8,1)-- (7,1)-- (6,1) -- (5,1)-- (4,1) -- (3,1) --(2,1) -- (1,1) -- (0,1) -- (0,0) ;
		
		\draw (1,0) --(1,1);
		\draw (2,0) -- (2,1);
		\draw (3,0) -- (3,1);
		\draw (4,0) -- (4,1);
		\draw (5,0) -- (5,1);
		\draw (6,0) -- (6,1);
		\draw (7,0) -- (7,1);
		\draw (8,0) -- (8,1);
			
	\end{tikzpicture}
	\end{center}
	\footnotesize{\caption{The partition $(d)$ on which $\pmp$ concentrates asymptotically}}
\label{fig:specialPart}
\end{figure}

\

\par\noindent By the discussion in Section 2 this translates to a convergence result on $\fM^{(n)}_d$:

\begin{corollary}
	\label{cor:Upstairs}
	 If $n\geq 2d$ then the measure $\Pmp$ on $\fM_d^{(n)}$ converges weakly, as $q\longrightarrow 0^+$, to the uniform measure $\nu$ on $\fM^{(n)}_{d,1}$, the set of single partitions $\mu^{(1)}$ of $n$ with colength $d$. That is,
	 \begin{align}
	 	\nu(A) & = \frac{\abs{A\cap\fM^{(n)}_{d,1} }}{\abs{\fM^{(n)}_{d,1}}}.
	 \end{align}
\end{corollary}

\begin{remark}
	Note that the remaining surfaces, with exactly three branch points, with profiles $(\mu^{(1)}, \mu, \nu)$,
 correspond to what is often referred to as  Belyi curves \cite{AC1,KZ, Z}.
\end{remark}



\subsection{Higher-order corrections}

We now consider higher order terms for the partition function and the quantum Hurwitz numbers.
The follwing gives the two leading terms in weighted sum of functions on $\fM_d^{(n)}$ with weights $W_{E'(q)}$ on $\mP_d$.

\begin{theorem}
	\label{thm:FullExpansion}
	For any function $g\colon\fM_d^{(n)}\longrightarrow\rr$ we have
	\begin{align}
		\sum_{k=1}^d \sum_{(\mu^{(1)},\ldots,\mu^{(k)})\in\fM_{d,k}^{(n)}}& g\rb{ \mu^{(1)}, \ldots, \mu^{(k)} } W_{E'(q)}\rb{ \mu^{(1)},\ldots,\mu^{(k)} }  \\
		 &= q^d \sum_{\ell^*(\mu^{(1)})=d} g\rb{\mu^{(1)}}  + q^{d+1} \sum_{\substack{\ell^*(\mu^{(1)})=d-1\\ \ell^*(\mu^{(2)})=1}} g\rb{\mu^{(1)},\mu^{(2)}}  + O\rb{q^{d+2}}.
	\end{align}
\end{theorem}

\par\noindent In particular,  we obtain the following leading terms in the $T=0$ expansion of
 the partition function. (Recall that $p(d)$ denotes the number of integer partitions of $d$.)

\begin{corollary}
	As $q\rightarrow 0^+$,
	\label{cor:PartitionFunction}
	\begin{align}
		\pfp & = p(d) q^d + p(d-1) q^{d+1} + O\rb{q^{d+2}}.
	\end{align}
\end{corollary}

For the zero temperature expansions of simple quantum Hurwitz numbers, we have the following leading terms

\begin{corollary}
	\label{cor:QuantumWeighted}
	For any $\mu,\nu\in\mP_n$ we have, as $q\rightarrow 0^+$,
	\begin{align}
		H^d_{E'(q)}(\mu,\nu)  = q^{d}\sum_{\substack{ \mu^{(1)} \in\fM_{d,k}^{(n)} \\  \ell^*(\mu^{(1)})=d}}
 H\rb{\mu^{(1)},\mu,\nu}+ q^{d+1}\sum_{\substack{ \mu^{(1)}, \,\mu^{(2)} \in\fM_{d,k}^{(n)} \\  \ell^*(\mu^{(1)})+  \ell^*(\mu^{(2)})=d}}  H\rb{\mu^{(1)},\mu^{(2)},\mu,\nu}+O\rb{q^{d+2}}.
	\end{align}
	\end{corollary}


\section{Proofs}
\label{sec:proofs}

In this section we prove the results stated in Section \ref{sec:scl}.

\begin{proof}[Proof of Theorem \ref{thm:Downstairs}]
	For any $\lambda\in\mP_d$,
	\begin{align}
		w_{E'(q)}(\lambda) & = \frac1{\aut(\lambda)} \sum_{\sigma \in S_{\ell(\lambda)}} \prod_{j=1}^{\ell(\lambda)} \frac{q^{\sum_{i=1}^j \lambda_{\sigma(i)} }}{ 1-q^{\sum_{i=1}^j \lambda_{\sigma(i)} } }\\
		& = \frac1{\aut(\lambda)} \sum_{\sigma \in S_{\ell(\lambda)}} \prod_{j=1}^{\ell(\lambda)} q^{\sum_{i=1}^j \lambda_{\sigma(i)} } \rb{1+O\rb{q}}\\
		& = \frac1{\aut(\lambda)} \sum_{\sigma \in S_{\ell(\lambda)}}  q^{\sum_{i=1}^{\ell(\lambda)} (\ell(\lambda)-i+1)\lambda_{\sigma(i)} } \rb{1+O\rb{q}}\\
		& = \frac1{\aut(\lambda)} \sum_{\sigma \in S_{\ell(\lambda)}}  q^{\sum_{j=1}^{\ell(\lambda)} j\lambda_{\sigma(j)} } \rb{1+O\rb{q}}
	\end{align}
	Since $\lambda_1\geq \lambda_2\geq\ldots\geq \lambda_{\ell(\lambda)}$ and $q$ is small the  sum above is dominated by the contribution when $\sigma$ is the identity permutation. In particular we obtain
	\begin{align}
		\label{eq:FirstOrderW}
		w_{E'(q)}(\lambda) & = \frac1{\aut(\lambda)}\,  q^{\sum_{j=1}^{\ell(\lambda)} j\lambda_{j} } \rb{1+O\rb{q}}.
	\end{align}
	Thus, in the limit as $q\to 0$ the dominant weight will be given by $\lambda$ such that $\sum_j j\lambda_j$ is minimised, i.e. $\lambda=(d)$. This completes the proof.
\end{proof}

\par\noindent Having established the first-order result we now turn to the higher order corrections. It follows from \eqref{eq:FirstOrderW} that
\begin{align}
	\sum_{\substack{\lambda\ne (d)\\\lambda\ne (d-1,1)}} w_{E'(q)}(\lambda ) & = O\rb{q^{d+2}}.
\end{align}
Further,
\begin{align}
	w_{E'(q)}((d)) & = \frac1{\aut((d))}\ \frac{q^d}{1-q^d} = q^d+O\rb{q^{2d}},\\
	w_{E'(q)}((d-1,1)) &= \frac{q^{d-1}}{1-q^d}\,\frac{q^d}{1-q^d} + \frac q{1-q}\, \frac{q^d}{1-q^d}\\
	& = \frac{q^{d+1}}{(1-q)(1-q^d)} + O\rb{q^{2d-1}}\\
	& = q^{d+1}\rb{1+q+O\rb{q^2}}\rb{1+q^d+O\rb{q^{2d}}} \\
	&= q^{d+1} +  q^{d+2} +O\rb{q^{d+3}}.
\end{align}
For any $f\colon\mP_d\longrightarrow\rr$ we therefore have
\begin{align}
	\label{eq:FullExpansionDownstairs}
		\sum_{\lambda\in\mP_d} f(\lambda) w_{E'(q)}(\lambda) & = f((d))\ p(d)  q^d + f((d-1),(1))p(d-1)  q^{d+1} + O\rb{q^{d+2}}.
\end{align}
Theorem \ref{thm:FullExpansion} now follows from \eqref{eq:FullExpansionDownstairs} and the discussion in Section 2. 

By choosing $g$ to be identically equal to 1 we obtain Corollary \ref{cor:PartitionFunction}, whereas choosing, for fixed $(\mu,\nu)$,
\begin{align}
	g(\mu^{(1)},\ldots,\mu^{k})=H(\mu^{(1)},\ldots,\mu^{k},\mu,\nu),
\end{align}   Corollary \ref{cor:QuantumWeighted} is just Theorem \ref{thm:FullExpansion} for this particular case.
\ \\ \ \\

\noindent 
\small{ {\it Acknowledgements.}  The work of JH was partially supported by the Natural Sciences and Engineering Research Council of Canada (NSERC) and the Fonds de recherche du Qu\'ebec, Nature et technologies (FRQNT). JO was partially supported by a CRM-ISM postdoctoral fellowship and a Concordia Horizon postdoctoral fellowship.}

\newcommand{\arxiv}[1]{\href{http://arxiv.org/abs/#1}{arXiv:{#1}}}

\bigskip
\noindent

\end{document}